\documentclass{acmsiggraph}	

\usepackage{wrapfig}
\usepackage{float}
\usepackage{pdfpages} 
\usepackage{amsthm}	
\usepackage{graphicx,wrapfig}
\usepackage{graphicx}
\usepackage{parskip}
\usepackage{enumerate}
\usepackage{microtype}
\usepackage{marginnote}
\usepackage{subfigure}
\usepackage{mathdots}
\usepackage[ruled,vlined,linesnumbered]{algorithm2e}
\usepackage{todonotes}

\usepackage{overpic}
\usepackage{array}
\usepackage{multirow}
\usepackage[toc,page]{appendix}

\usepackage{amsmath,amssymb,amsfonts,mathtools}
\usepackage{amsthm}
\newtheorem{theorem}{Theorem}[section]
\newtheorem*{theorem*}{Theorem}
\newtheorem{lemma}[theorem]{Lemma}
\newtheorem*{lemma*}{Lemma}

\newtheorem*{proposition*}{Proposition}
\newtheorem{corollary}[theorem]{Corollary}
\newtheorem*{corollary*}{Corollary}
\renewenvironment{proof}[1][Proof]{\begin{trivlist}
\item[\hskip \labelsep {\bfseries #1}]}{\hfill{$\square$}\end{trivlist}}
\newenvironment{definition}[1][Definition]{\begin{trivlist}
\item[\hskip \labelsep {\bfseries #1}]}{\end{trivlist}}
\newenvironment{notation}[1][Notation]{\begin{trivlist}
\item[\hskip \labelsep {\bfseries #1}]}{\end{trivlist}}
\newenvironment{remark}[1][Remark]{\begin{trivlist}
\item[\hskip \labelsep {\bfseries #1}]}{\end{trivlist}}

\def\BibTeX{{\rm B\kern-.05em{\sc i\kern-.025em b}\kern-.08em
    T\kern-.1667em\lower.7ex\hbox{E}\kern-.125emX}}

\TOGonlineid{0248}
\TOGvolume{0}
\TOGnumber{0}
\title{Variance Analysis for Monte Carlo Integration:\\ A Representation-Theoretic Perspective}

\author{
Michael Kazhdan$^1$
\hspace*{0.5cm}
Gurprit Singh$^{2}$
\hspace*{0.5cm}
Adrien Pilleboue$^{2}$
\hspace*{0.5cm}
David Coeurjolly$^3$
\hspace*{0.5cm}
Victor Ostromoukhov$^{2,3}$
\\
	\hspace*{0.5cm}
	$^1$Johns Hopkins University
	\hspace*{0.5cm}
	$^2$Universit\'e Lyon 1
	\hspace*{0.5cm}
	$^3$CNRS/LIRIS UMR 5205
}
\pdfauthor{Kazhdan et al.}

\renewcommand{\P}{{\mathcal P}}
\newcommand{\C}{{\mathbb C}}
\newcommand{\R}{{\mathbb R}}
\newcommand{\Z}{{\mathbb Z}}
\renewcommand{\S}{{\mathbf S}}

\renewcommand{\v}{{\mathbf v}}
\renewcommand{\d}[1]{{\ \mathsf{d}#1}}
\newcommand{\Mean}[2]{\mathsf{E}_{#2}\!\left[#1\right]}
\newcommand{\Var}[2]{\mathsf{Var}_{#2}\!\left(#1\right)}
\newcommand{\G}{{\Gamma}}
\newcommand{\g}{{\gamma}}
\newcommand{\Id}{\hbox{Id.}}
\newcommand{\minusqquad}{\mkern-18mu}
\newcommand{\Int}[1]{\underset{{#1}}{\mathop\int}}
\newcommand{\ignore}[1]{}
\newcommand\CommentMK[1]{{\color{violet!80}\textbf{[MK]} #1}}

\begin{document}
\maketitle


\section{Overview}
In this report, we revisit the work of Pilleboue et al.~\shortcite{Pilleboue:SIGGRAPH:2015}, providing a representation-theoretic derivation of the closed-form expression for the expected value and variance in homogeneous Monte Carlo integration. We show that the results obtained for the variance estimation of Monte Carlo integration on the torus, the sphere, and Euclidean space can be formulated as specific instances of a more general theory. We review the related representation theory and show how it can be used to derive a closed-form solution.

\section{Problem Statmement}
\label{sec:problem_statement}
We begin by reviewing some basic concepts from Monte Carlo integration. Next, we present a formal definition of homogeneity. And finally, we formulate the generalized problem statement.

\subsection*{Monte Carlo Integration}

\begin{definition}
Given a domain $\Omega$ and given two (complex-valued) functions $F,G:\Omega\rightarrow\C$, the \emph{dot-product} of the functions is the integral of the product of $F$ with the complex conjugate of $G$:
$$\langle F, G\rangle = \Int{\Omega} F(x)\cdot \overline{G(x)}\d{x}.$$
\end{definition}

\begin{definition}
Given a domain $\Omega$ and given $S=\{s_1,\cdots,s_N\}\in\Omega^N$, the \emph{Monte Carlo estimate} of the integral of a function $F:\Omega\rightarrow\C$ is obtained by averaging the values of $F$ at the $N$ positions:
$$\hbox{MC}(F,S) := \frac1N \sum_{i=1}^N F(s_i).$$
Treating $S$ as the average of delta functions, centered at $\{s_i\}$:
$$S(x)\equiv\frac1N\sum_{i=1}^N \delta_{s_i}(x),$$
the Monte Carlo estimate becomes the dot-product of $F$ and $S$:
$$\hbox{MC}(F,S)=\left\langle F , S\right\rangle.$$
\end{definition}

\begin{definition}
Given a domain $\Omega$ and a positive integer $N$, a \emph{sampling pattern} is a function $\P:\Omega^N\rightarrow\R$ over the set of all $N$-tuples of points in $\Omega$, satisfying:
$$\Int{\Omega^N} \P(S)\ \mathsf{d}S=|\Omega|\qquad\hbox{and}\qquad \P(S)\geq0,\quad\forall\ S\in\Omega^N,$$
where $|\Omega|$ is the measure of $\Omega$.
\end{definition}

\begin{definition}
Given a sampling pattern $\P$ and a function $F:\Omega\rightarrow\C$, the \emph{expected value} of the integral of $F$ and the \emph{variance} in the estimate of the integral are given by:
\begin{align*}
\Mean{\langle F,\S\rangle}{\P} &:= \Int{\Omega^N} \langle F,S\rangle\cdot \P(S)\d{S}\\
\Var{\langle F,\S\rangle}{\P} &:= \Mean{\left\|\langle F,\S\rangle\right\|^2}{\P} - \left\|\Mean{\langle F,\S\rangle}{\P}\right\|^2.
\end{align*}
\end{definition}

\subsection*{Homogeneity}
In order to make the problem of estimating the variance in Monte Carlo integration tractable, we restrict ourselves to sampling patterns that are homogeneous. To make this formal, we first define a notion of a group action.

\begin{definition}
We say that a group, $\G$, \emph{acts on} $\Omega$ if each element $\g\in\G$ defines a maps $\g:\Omega\rightarrow\Omega$ the preserves the measure on $\Omega$.
\end{definition}

\begin{notation}
Given a group action of $\G$ on $\Omega$, given $F:\Omega\rightarrow\C$, and given $\g\in\G$, we denote by $\g(F):\Omega\rightarrow\C$ the function obtained by applying the inverse of $\g$ to the argument of $F$:
$$[\g(F)](x) := F\left(\g^{-1}(x)\right).$$
Here, inversion is required so that $(\g\circ \tilde{\g})(F) = \g(\tilde{\g}(F))$ for all $\g,\tilde{\g}\in\G$.
\end{notation}

\begin{remark}
Since the map $\g:\Omega\rightarrow\Omega$ preserves the measure, the associated map on the space of functions is unitary. That is, for any functions $F,G:\Omega\rightarrow\C$ we have:
$$\left\langle F , G \right\rangle = \left\langle \g(F) , \g(G)\right\rangle,\qquad\forall\ \g\in\G.$$
\end{remark}

\begin{definition}
Given a group action of $\G$ on $\Omega$, we say that a sampling pattern $\P:\Omega^N\rightarrow\R$ is \emph{homogeneous with respect to $\G$} if the probability of choosing a sampling pattern is the same as the probability of choosing any of its transformations by the group elements:
$$\P(S) = \P\left(\g(S)\right),\qquad\forall\ \g\in\G.$$
(Note that we can either think of $\g(S)$ as the sampling pattern obtained by transforming the sample positions, $s_i\mapsto\g(s_i)$, or as the transform of the sum of delta functions -- the two definitions are consistent.) 
\end{definition}

\begin{remark}
If the group $\G$ is compact, one can always transform an initial sampling pattern $\P_0$ into a homogeneous sampling pattern $\P$ by averaging over the group elements:
$$\P(S):=\frac{1}{|\G|}\Int{\G}\P_0\left(\g(s)\right)\d{\g}.$$
\end{remark}

\begin{remark}
It is common to use the term \emph{homogeneous} to refer to invariance to translation and the term \emph{isotropic} to refer to invariance to rotations. As the general theory we present will not distinguish between the group actions, we will use the term \emph{homogeneous} throughout.
\end{remark}

\subsection*{Problem Statement}
Thinking of the space of functions as a complex inner-product space, thinking of a sampling pattern as a real-valued function on this vector space, and using the fact that computing the Monte Carlo integral amounts to taking the dot-product of the integrand with the average of delta functions, we can view the problem of estimating variance in Monte Carlo integration as an instance of the following, more general, algebraic problem:

\emph
{
Assume we are given a complex inner-product space $(V,\langle\cdot,\cdot\rangle)$, a group $\G$ acting on $V$, and a homogeneous function $\P:V\rightarrow\R$. Then, given $w\in V$, compute the expected value and variance of the dot-product of $w$ with the vectors $\v\in V$:
\begin{align*}
\Mean{\langle w ,\v\rangle}{\P} &= \Int{V}\langle w , v\rangle \cdot \P(v)\d{v}\\
\Var{\langle w ,\v\rangle}{\P} &= \Mean{\left\|\langle w , \v\rangle \right\|^2}{\P}-\left\|\Mean{\langle w , \v\rangle }{\P}\right\|^2.
\end{align*}
}

The advantage of formulating the problem in this manner is that it makes it easier to leverage representation theory to find a solution. To this end, we review some basic concepts from representation theory in the next section, as well as derive two lemmas describing how the average of the inner-products of vectors behave as we transform one of the arguments by the elements of the group. Using these, we present our closed-form expression for the expected value and variance, given in terms of the Fourier coefficients of the integrand $F$ and the sampling patterns $\S$, in Section~\ref{sec:variance_estimation}.


\section{Representation Theory}
\label{sec:representation_theory}
The study of how the Fourier coefficients of a signal change as it is transformed by the elements of a group is best expressed in the language of representation theory. We review some basic concepts from this theory, before deriving the lemmas that lead to a closed-form expression for the expected value and variance of the Monte Carlo integral.

In what follows, we will assume a compact (closed and bounded) Lie group $\G$.

\begin{definition}
Given a complex inner-product space $(V,\langle\cdot,\cdot\rangle)$, we say that the $(\rho,V)$ is a \emph{representation of $\G$} if $\rho$ is a group homomorphism from $\G$ into the group of unitary transformations on $V$. That is:
$$\rho(\g\circ\tilde{\g})=\rho(\g)\circ\rho(\tilde{\g}),\qquad\forall\ \g,\tilde{\g}\in\G.$$
\end{definition}

\begin{notation}
Given a representation $(\rho,V)$, a group element $\g\in\G$, and a vector $v\in V$, we will write:
$$\g(v):=\rho(\g)(v).$$
\end{notation}

\begin{definition}
Given a vector space $V$, the \emph{trivial representation} is the map $\rho$ sending every group element to the identity:
$$\rho(\g)=\Id,\qquad\forall\ \g\in\G.$$
\end{definition}

\begin{definition}
Given a representation $(\rho,V)$ and a subspace $W\subset V$, we say that $W$ is a {\em sub-representation} if $\g(w)\in W$ for all $w\in W$ and all $\g\in\G$.
\end{definition}

\begin{definition}
We say that $(\rho,V)$ is an {\em irreducible representation} if the only sub-representations are $W=\{0\}$ and $W=V$.
\end{definition} 

\begin{remark}
Since any subspace of a trivial representation is a sub-representation (as the identity maps all vectors to themselves), a trivial representation is irreducible if and only if it is one-dimensional.
\end{remark}

Given a representation $(\rho,V)$, Maschke's Theorem~\cite{Serre:1977,Fulton:1991} tells us that we can decompose $V$ as the direct sum of finite-dimensional, irreducible representations:
$$V=\bigoplus_{\lambda\in\Lambda}V^\lambda$$
with $V^\lambda$ and $V^{\tilde{\lambda}}$ perpendicular whenever $\lambda\neq\tilde{\lambda}$.

Choosing an orthonormal basis $\{b^1_\lambda,\cdots,b^{n_\lambda}_\lambda\}$ for each $V^\lambda$ allows us to define a Fourier transform:

\begin{definition}
Given a vector $v\in V$, and indices $\lambda\in\Lambda$ and $m\in[1,\cdots,n_\lambda]$, the \emph{$(\lambda,m)$-th Fourier coefficient} of $v$, denoted $\widehat{v}_\lambda^m$, is the coefficient of $v$ corresponding to the basis vector $b_\lambda^m$:
$$\widehat{v}_\lambda^m:=\langle v,b_\lambda^m\rangle.$$
\end{definition}

\begin{remark}
Since the basis defining the Fourier coefficients is orthonormal, we can write the inner product of two functions $v,w\in V$ in terms of these coefficients as:
$$\langle v, w\rangle \equiv \sum_{\lambda\in\Lambda}\sum_{m=1}^{n_\lambda} \widehat{v}_\lambda^m\cdot\overline{\widehat{w}_\lambda^m}.$$
\end{remark}

\begin{lemma}
\label{l:lemma1}
Given an irreducible representation $(\rho,V)$ of a group $\G$, for any $x,y,v,w\in V$, we have:
$$\frac{1}{|\G|}\int_\G\langle\g(x),y\rangle \cdot\overline{\langle\g(v),w\rangle}\d{\g} = 
\frac{1}{\hbox{dim}(V)}\cdot\langle x , v\rangle\cdot \overline{\langle y , w\rangle}.$$
\end{lemma}

\begin{corollary}
In particular, letting $\{b^1,\cdots,b^n\}$ be an orthonormal basis for $V$, taking $y=b^i$ and $w=b^j$, and fixing $x=v$, the above statement becomes:
$$\int_\G \widehat{\g(v)}^i\cdot\overline{\widehat{\g(v)}^j} = \frac{|\G|}{\hbox{dim}(V)}\cdot\|v\|^2\cdot\delta_{ij}.$$
That is, the Fourier coefficients of $\g(v)$, thought of as complex-valued functions on $\G$, are orthogonal and the magnitude is independent of which Fourier coefficient we are considering.
\end{corollary}

\begin{proof}
Fixing $y,w\in V$, let $B_{y,w}:V\times V\rightarrow\C$ be the map:
$$B_{y,w}(x,v)=\int_\G\langle\g(x),y\rangle\cdot\overline{\langle\g(v),w\rangle}\d{\g}.$$
It is not hard to show that this map is linear in the first argument, conjugate-linear in the second, and $\G$-equivariant. (That is, for any $\g\in\G$ we have $B_{v,w}\left(\g(x),\g(y)\right)=B_{v,w}(x,y)$).
Thus, by Schur's Lemma~\cite{Serre:1977,Fulton:1991}, $B_{y,w}$ is a scalar multiple of the inner-product on $V$:
$$B_{y,w}(x,v)=\lambda_{y,w}\cdot\langle x , v \rangle.$$
Noting that $B_{y,w}(x,v)=\overline{B_{x,v}(y,w)}$, it follows that:
\ignore
{
\newline
{\bf [MK]This should be commented out}
\begin{align*}
\lambda_{y,w}\cdot\langle x , v \rangle &= \overline{\lambda_{x,v}\cdot\langle y , w \rangle }\\
&\Updownarrow\\
\frac{\lambda_{y,w}}{\overline{\langle y , w \rangle}} &= \frac{\overline{\lambda_{x,v}}}{\langle x , v \rangle}\\
\end{align*}
}
$$B_{y,w}(x,v) = \lambda\cdot\langle x , v \rangle\cdot\overline{\langle y , w\rangle},$$
for some constant $\lambda\in{\mathbb C}$ that is independent of $v$ and $w$.

Thus, we are left with the problem of determining $\lambda$. As it is independent of $x$, $y$, $v$, and $w$, it suffices to determine the value $B_{v,v}(v,v)$ for some $v\neq0$.
More generally, letting $\{b^1,\ldots,b^n\}$ be an orthonormal basis we can get an expression for $\lambda$ in terms of the integrated square norm of the trace of $\rho(\g)$:
\begin{eqnarray*}
\int_\G\left\|\hbox{Tr}\left(\rho(\g)\right)\right\|^2\d{\g}
&=& \int_\G\left\|\sum_{i,j=1}^n \langle \g(b^i),b^j\rangle\right\|^2\d{\g}\\
&=& \sum_{i=1}^n B_{b^i,b^j}(b^i,b^j)\\
&=& \hbox{dim}(V)\cdot\lambda.
\end{eqnarray*}
Since the trace is the character of the representation, it follows by the orthogonality of characters~\cite{Serre:1977,Fulton:1991} that $\int_\G\left\|\hbox{Tr}\left(\rho(\g)\right)\right\|^2\d{\g} = |\G|$, giving $\lambda = |\G|/\hbox{dim}(V)$.

Thus, as desired, we get:
$$\frac{1}{|\G|}\int_\G\langle\g(x),y\rangle\cdot\overline{\langle\g(v),w\rangle}\d{\g} = 
\frac{1}{\hbox{dim}(V)}\cdot\langle x , v\rangle\cdot\overline{\langle y , w\rangle}.$$
\end{proof}

\begin{lemma}
\label{l:lemma2}
Leveraging Schur's Lemma in a similar manner, it follows that if $(\rho_1,V_1)$ and $(\rho_2,V_2)$ are two irreducible representations that are not isomorphic, then for any $v_1,w_1\in V_1$ and $v_2,w_2\in V_2$:
$$\int_\G\langle\g(v_1),w_1\rangle\cdot\overline{\langle\g(v_2),w_)\rangle}\d{\g} = 0.$$
\end{lemma}

\section{Variance Estimation}
\label{sec:variance_estimation}
Using the above theory, we are now prepared to estimate the variance in Monte Carlo integration. We begin by presenting a general expression for the expected value and variance and then consider the specific cases of the torus, the sphere, and Euclidean space.

\subsection{General Framework}
We assume that we are given a representation $(\rho,V)$ of a compact group $\G$, a homogeneous function $\P:V\rightarrow\R$ (i.e. $\P(v)=\P(\g(v))$ for all $v\in V$ and all $\g\in\G$), and a vector $w\in V$. Our goal is to express the expected value and variance of the dot-product of $w$ with the vectors $\v\in V$:
\begin{align*}
\Mean{\langle w,\v\rangle}{\P} &= \Int{V} \langle w,v\rangle\cdot \P(v)\d{v}\\
\Var {\langle w,\v\rangle}{\P} &= \Mean{\left\|\langle w,\v\rangle\right\|^2}{\P} - \left\|\Mean{\langle w,\v\rangle}{\P}\right\|^2.
\end{align*}

In deriving the closed form expression for the expected value and variance, we will assume that the decomposition of $V$ into irreducible representations $\{V^\lambda\}_{\lambda\in\Lambda}$ contains the trivial representation. (If it does not, we can take the direct sum of $V$ with a one-dimensional space on which $\G$ acts trivially.) We will denote this one-dimensional representation as $V^0$ and let $b_0^1$ be a unit-vector spanning this space (with $\g(b_0^1)=b_0^1$ for all $\g\in\G$).

Finally, for simplicity, we will assume that the irreducible representations occur without multiplicity. That is, if $\lambda\neq\widetilde{\lambda}$ then $V^\lambda$ and $V^{\widetilde{\lambda}}$ are not isomorphic, for all $\lambda,\widetilde{\lambda}\in\Lambda$.

\subsubsection*{The Expected Value}
Using the Fourier coefficients, we can express expected value of the dot-product of $w$ with the vectors $\v\in V$ as:
$$\Mean{\langle w,\v\rangle}{\P} = \Int{V}\sum_{\lambda\in\Lambda}\sum_{m=1}^{n_\lambda} \widehat{w}_\lambda^m\cdot\overline{\widehat{v}_\lambda^m}\cdot\P(v)\d{v}.$$

Using the homogeneity of $\P$, the expected value computed by integrating over $V$ is the same as the expected value computed by integrating over $\g(V)$. In particular, we can express the expected value as the average:
\begin{align*}
\Mean{\langle w,\v\rangle}{\P}
&= \frac{1}{|\G|}\Int{\G}\Int{\g(V)}\sum_{\lambda\in\Lambda}\sum_{m=1}^{n_\lambda} \widehat{w}_\lambda^m\cdot\overline{\widehat{v}_\lambda^m}\cdot\P(v)\d{v}\d{\g}\\
&\minusqquad\minusqquad= \frac{1}{|\G|}\Int{\G}\Int{V}\sum_{\lambda\in\Lambda}\sum_{m=1}^{n_\lambda} \widehat{w}_\lambda^m\cdot\overline{\widehat{\g(v)}_\lambda^m}\cdot\P(\g(v))\d{v}\d{\g}\\
&\minusqquad\minusqquad= \Int{V}\sum_{\lambda\in\Lambda}\sum_{m=1}^{n_\lambda}\widehat{w}_\lambda^m\frac{1}{|\G|}\left(\Int{\G}\overline{\langle \g(v) , b_\lambda^m\rangle}\d{\g}\right)\cdot\P(v)\d{v}.
\end{align*}
Since $b_0^1$ is a unit vector on which $\g$ acts as the identity, we have:
\begin{align*}
\Mean{\langle w,\v\rangle}{\P}
&= \Int{V}\sum_{\lambda\in\Lambda}\sum_{m=1}^{n_\lambda}\widehat{w}_\lambda^m\cdot\\
&\cdot\frac{1}{|\G|}\left(\Int{\G}\overline{\langle\g(v),b_\lambda^m\rangle\cdot\overline{\langle\g(b_0^1),b_0^1\rangle}}\d{\g}\right)\cdot\P(v)\d{v}.
\end{align*}
Using Lemmas~\ref{l:lemma1} and~\ref{l:lemma2}, the facts that $V^0$ is orthogonal and not isomorphic to $V^\lambda$ for all $\lambda\neq0$, and that $\hbox{dim}(V^0)=1$, we have:
\begin{align*}
\Mean{\langle w,\v\rangle}{\P}
&= \sum_{\lambda\in\Lambda}\sum_{m=1}^{n_\lambda}\Int{V}\widehat{w}_\lambda^m\cdot\overline{\langle v,b_0^1\rangle}\cdot\langle b_\lambda^m,b_0^1\rangle\cdot\P(v)\d{v}\\
&= \Int{V}\widehat{w}_0^1\cdot\overline{\widehat{v}_0^1}\cdot\P(v)\d{v}\\
&= \widehat{w}_0^1\cdot\overline{\Int{V}\widehat{v}_0^1\cdot\P(v)\d{v}}.
\end{align*}
That is, the expected value of the dot-product is the trivial Fourier coefficient of $w$ times the complex conjugate of the expected value of the trivial Fourier coefficient of the vectors $\v\in V$:
\begin{equation}
\label{e:expected}
\boxed{\Mean{\langle w,\v\rangle}{\P} = \widehat{w}_0^1\cdot\overline{\Mean{\widehat{\v}_0^1}{\P}}.}
\end{equation}

\subsection*{The Variance}
Using Equation~(\ref{e:expected}), we can express the variance of the dot-product of $w$ with the vectors $\v\in V$ as:
$$\Var{\langle w,\v\rangle}{\P} = \Mean{\left\|\langle w,\v\rangle\right\|^2}{\P}-\left\|\widehat{w}_0^1\right\|^2\cdot\left\|\Mean{\widehat{\v}_0^1}{\P}\right\|^2$$
and we are left with the problem of computing $\Mean{\left\|\langle w,\v\rangle\right\|^2}{\P}$.

Expressing the dot-product in terms of the Fourier coefficients gives:
\begin{align*}
\Mean{\left\|\langle w,\v\rangle\right\|^2}{\P}
&= \Int{V}\left\|\sum_{\lambda\in\Lambda}\sum_{m=1}^{n_\lambda}\widehat{w}_\lambda^m\cdot\overline{\widehat{v}_\lambda^m}\right\|^2\cdot\P(v)\d{v}\\
&\minusqquad\minusqquad= \Int{V}\sum_{\lambda,\widetilde{\lambda}\in\Lambda}\sum_{m=1}^{n_\lambda}\sum_{\widetilde{m}=1}^{n_{\widetilde{\lambda}}}
\widehat{w}_\lambda^m\cdot\overline{\widehat{w}_{\widetilde{\lambda}}^{\widetilde{m}}}\cdot\overline{\widehat{v}_\lambda^m}\cdot\widehat{v}_{\widetilde{\lambda}}^{\widetilde{m}}\cdot\P(v)\d{v}.
\end{align*}

As with the expected value, homogeneity implies that we can average the integrals over all $\gamma(V)$, giving:
\begin{align*}
\Mean{\left\|\langle w,\v\rangle\right\|^2}{\P}
&= \Int{V}\sum_{\lambda,\widetilde{\lambda}\in\Lambda}\sum_{m=1}^{n_\lambda}\sum_{\widetilde{m}=1}^{n_{\widetilde{\lambda}}}
\widehat{w}_\lambda^m\cdot\overline{\widehat{w}_{\widetilde{\lambda}}^{\widetilde{m}}}\cdot\\
&\qquad\cdot
\frac{1}{|\G|}\left(\Int{\G}\overline{\widehat{\g(v)}_\lambda^m}\cdot\widehat{\g(v)}_{\widetilde{\lambda}}^{\widetilde{m}}\d{\g}\right)\cdot\P(v)\d{v}.
\end{align*}
Using the fact that $\widehat{\g(v)}_\lambda^m=\langle\g(v),b_\lambda^m\rangle$ in conjunction with Lemmas~\ref{l:lemma1} and~\ref{l:lemma2} and letting $\pi_\lambda:V\rightarrow V^\lambda$ be the projection from $V$ onto the irreducible representation $V^\lambda$, the summation simplifies to:
\begin{align*}
\Mean{\left\|\langle w,\v\rangle\right\|^2}{\P}
&= \Int{V}\sum_{\lambda\in\Lambda}\sum_{m=1}^{n_\lambda}\left\|\widehat{w}_\lambda^m\right\|^2\cdot\frac{\left\|\pi_\lambda(v)\right\|^2}{\hbox{dim}(V^\lambda)}\cdot\P(v)\d{v}\\
&= \sum_{\lambda\in\Lambda}\frac{\left\|\pi_\lambda(w)\right\|^2\cdot\Mean{\left\|\pi_\lambda(\v)\right\|^2}{\P}}{\hbox{dim}(V^\lambda)}.
\end{align*}
This gives a closed-form expression for the variance as:
\begin{equation}
\boxed{
\label{e:variance}
\Var{\langle w,\v\rangle}{\P}
= \sum_{\lambda\in\Lambda\setminus\{0\}}\frac{\left\|\pi_\lambda(w)\right\|^2\cdot\Mean{\left\|\pi_\lambda(\v)\right\|^2}{\P}}{\hbox{dim}(V^\lambda)}.
}
\end{equation}
Note that by taking the summation over all irreducible representations \emph{except for the trivial one} we subtract off the square-norm of the expected value.

\subsection{The Torus}
In this case, the domain of integration and the group of motions are both the $d$-dimensional torus, $\Omega = \G = [0,2\pi)^d$, and the representation is defined on the space of complex-valued functions on the torus, $V=L^2(\Omega,\C)$, with an element $\g\in\G$ acting on a function by translation:
$$[\g(F)](p):=F(p-\g).$$

The irreducible representations are all one-dimensional (since the group is commutative) and are indexed by points on the $d$-dimensional integer lattice, $\Lambda = \Z^d$. Specifically, the space $V^\lambda$ is spanned by a complex exponential with frequency $\lambda\in\Z^d$:
$$V^\lambda = \hbox{Span}\left\{b_\lambda^0(p)=\frac{e^{i\langle p , \lambda\rangle}}{(2\pi)^{d/2}}\right\}.$$

Thus, Equation~(\ref{e:expected}) gives the expected value of the integral of $F$ as the product of the DC component of $F$ times the complex conjugate of the expected value of the DC component of $\S$. As we are considering Monte Carlo integration, the elements of $\S$ are all the average of $N$ delta-functions, so that:
$$\Mean{\langle F,\S\rangle}{\P}=\Int{\Omega}F(p)\d{p},$$
and the estimate is unbiased.

From Equation~(\ref{e:variance}) the variance in the estimate of the integral can be obtained by taking the power spectrum of $F$, multiplying (frequency-wise) by the expected power spectrum of $\S$, and summing over all non-zero frequencies:
$$\Var{\hbox{MC}(F,\S)}{\P} = \sum_{\lambda\in\Z^d\setminus\{0\}}\|\widehat{F}_\lambda\|^2\cdot\Mean{\|\widehat{\S}_\lambda\|^2}{\P},$$
where $\widehat{F}_l$ is the $l$-th Fourier coefficient of $F$.

\subsection{The Sphere}
In this case, the domain of integration is the 2-sphere, $\Omega=S^2$, the group of motions is the group of rotations in 3D, $\G=SO(3)$, and the representation is defined on the space of complex-valued functions on the sphere, $V=L^2(\Omega,\C)$, with an element $\g\in\G$ acting on a function by rotation:
$$[\g(F)](p):=F\left(\g^{-1}(p)\right).$$
In this case, the irreducible representations are indexed by the non-negative integers, $\Lambda = [0,\cdots,\infty)$, and the irreducible representation $V^\lambda$ is a $(2\lambda+1)$-dimensional space: $$V^\lambda = \hbox{Span}\left\{Y_\lambda^{-\lambda}(\theta,\phi),\cdots,Y_\lambda^\lambda(\theta,\phi)\right\},$$
with $Y_l^m(\theta,\phi)$ the spherical harmonic of frequency $l$ and index $m$.

As with the torus the integrator is unbiased, and the variance can be computed by summing, over each non-zero spherical frequency, the product of the power of $F$ and  the expected power of $\S$ in that frequency, divided by the dimension of the frequency space:
$$\Var{\hbox{MC}(F,\S)}{\P} = \sum_{l=1}^\infty\frac{\displaystyle\sum_{m=-l}^l\|\widehat{F}_l^m\|^2\cdot\sum_{m=-l}^l\Mean{\|\widehat{\S}_l^m\|^2}{\P}}{2l+1},$$
where $\widehat{F}_l^m$ is the $(l,m)$-th spherical harmonic coefficient of $F$.

\subsection{Euclidean Space}
In this case, the domain of integration is $d$-dimensional Euclidean space, $\Omega=\R^d$, the group is the group of Euclidean motions, $\G=SE(d)=\R^d\times SO(d)$, and the representation is defined on the space of complex-valued functions on Euclidean space, $V=L^2(\Omega,\C)$, with an element $\g = (\tau,\sigma)\in\G$ acting on a function by a combination of translation and rotation:
$$[\g(F)](p):=F\left(\sigma^{-1}(p-\tau)\right).$$

Unfortunately, the analysis in Section~\ref{sec:representation_theory} does not apply to this context because we assumed that the group is compact. None-the-less, we can formally carry over the results, replacing the notion of ``dimension'' with the the ``size'' of the irreducible representations.

In this case, the irreducible representations are indexed by the non-negative real numbers, $\Lambda = \R^{\geq 0}$ \cite{Vilenkin:1978} and the space $V^\lambda$ is the ``span'' of complex exponentials whose frequency has norm $\lambda$:
$$V^\lambda = \hbox{Span}_{|q|=\lambda}\left\{b_\lambda^q(p)=e^{i\langle p , q\rangle}\right\}.$$

As above, the integrator is unbiased and, using the fact that the size of the $\lambda$-th irreducible representation is the size of of the $(d-1)$-dimensional sphere with radius $\lambda$, we get:
\begin{align*}
\Var{\hbox{MC}(F,\S)}{\P}
&= \\
&\minusqquad\minusqquad\minusqquad\minusqquad\minusqquad\minusqquad= \int_{0}^\infty\frac{\displaystyle
\Int{|q|=\lambda}\|\widehat{F}_{q}\|^2\d{q}\ \cdot\!\Int{|q|=\lambda}\Mean{\|\widehat{\S}_{q}\|^2}{\P}\d{q}}{\lambda^{d-1}\cdot|S^{d-1}|}\d{\lambda}-\\
&\qquad\qquad-\|\widehat{F}_0\|^2\cdot\Mean{\|\widehat{\S}_0\|^2}{\P},
\end{align*}
where $\widehat{F}_{p}$ is the $p$-th Fourier coefficient of $F$.

\bibliographystyle{acmsiggraph}
\bibliography{TechReport}
\end{document}